\documentclass[10pt,aps,prx,twocolumn,nofootinbib,showkeys,superscriptaddress]{revtex4-2}
\usepackage[utf8]{inputenc}
\usepackage[spanish, english]{babel}
\usepackage[ruled, noend, linesnumbered]{algorithm2e}
\SetAlCapFnt{\small}
\usepackage{comment}
\usepackage{dcolumn}
\usepackage{booktabs}
\usepackage{siunitx}
\usepackage{soul}
\usepackage[caption=false, position=top]{subfig}
\usepackage{amsfonts, mathtools, bm, amsthm}
\usepackage{xcolor}
\usepackage[breaklinks=true,colorlinks,citecolor=blue,linkcolor=blue,urlcolor=blue]{hyperref}
\usepackage{orcidlink}

\definecolor{Green}{RGB}{25, 186, 0}

\newtheorem{theorem}{Proposition}
\newtheorem{corollary}{Corollary}[theorem]

\newcommand{\trace}[1]{\mathrm{tr}\!\left[#1\right]}
\newcommand{\alglabel}[1]{Alg.~\ref{#1}}
\newcommand{\applabel}[1]{App.~\ref{#1}}
\newcommand{\figlabel}[1]{Fig.~\ref{#1}}
\newcommand{\seclabel}[1]{Sec.~\ref{#1}}
\newcommand{\tablabel}[1]{Tab.~\ref{#1}}
\newcommand{\theolabel}[1]{Th.~\ref{#1}}

\SetAlCapHSkip{0pt}% set caption margins
\makeatletter
\newcommand{\thealgorithm}{\arabic\algocf@float}

\newcommand{\AlgoCaptionFormat}{}

\renewcommand{\algocf@makecaption@ruled}[2]{%
  \global\sbox\algocf@capbox{\hskip\AlCapHSkip%
    \setlength{\hsize}{\columnwidth}% restored on exit of sbox
    \addtolength{\hsize}{-2\AlCapHSkip}% add equal margin to both sides
    \vtop{\AlgoCaptionFormat\algocf@captiontext{#1}{#2}}}% then caption is not centered
}%
\makeatother

\makeatletter
\def\bibsection{%
   \par
   \begingroup
    \baselineskip26\p@
    \bib@device{\hsize}{72\p@}%
   \endgroup
   \nobreak\@nobreaktrue
   \addvspace{19\p@}%
  }%
\makeatother

\begin{document}

\title{PauliComposer: Compute Tensor Products of Pauli Matrices Efficiently}
\author{Sebastián V. Romero\,\orcidlink{0000-0002-4675-4452}}
\email[]{sebastian.vidal@tecnalia.com}
\affiliation{TECNALIA, Basque Research and Technology Alliance (BRTA), 48160 Derio, Spain}
\affiliation{Department of Physical Chemistry, University of the Basque Country UPV/EHU, Apartado 644, 48080 Bilbao, Spain}
\author{Juan Santos-Suárez\,\orcidlink{0000-0001-9360-2411}}
\email[]{juansantos.suarez@usc.es}
\affiliation{Instituto Galego de Física de Altas Enerxías (IGFAE), Universidade de Santiago de Compostela, 15705 Santiago de Compostela, Spain}
\date{\today}

\begin{abstract}
    We introduce a simple algorithm that efficiently computes tensor products of Pauli matrices. This is done by tailoring the calculations to this specific case, which allows to avoid unnecessary calculations. The strength of this strategy is benchmarked against state-of-the-art techniques, showing a remarkable acceleration. As a side product, we provide an optimized method for one key calculus in quantum simulations: the Pauli basis decomposition of Hamiltonians.
\end{abstract}


\maketitle

\section{Introduction}\label{sec:intro}

Pauli matrices~\cite{Pauli_1927} are one of the most important and well-known set of matrices within the field of quantum physics. They are particularly important both in physics and chemistry when used to describe Hamiltonians of many-body spin glasses~\cite{Heisenberg_1928, Bethe_1931, Sherrington_Kirkpatrick_75, Panchenko_2012, Hubbard_1963, altland_simons_2006} or for quantum simulations~\cite{Jordan_1928, BRAVYI2002210, Seeley_2012, Tranter_2015, doi:10.1021/acs.jctc.8b00450, Steudtner_2018}. The vast majority of these systems are out of analytic control so that they are usually simulated through exact diagonalization which requires their Hamiltonians to be written in its matrix form. While this task may be regarded as a trivial matter in a mathematical sense, it involves the calculation of an exponentially growing number of operations. Furthermore, description of quantum systems via Matrix Product States (MPS)~\cite{ostlund1995thermodynamic}, Density Matrix Renormalization Group (DMRG)~\cite{white1992densitymatrix} and Projected Entangled Pair States (PEPS)~\cite{verstraete2004renormalization} also involve large scale Hamiltonians, as well as Lanczos method~\cite{lanczos1950iterationmethod}, whose formulation has been efficiently encoded on quantum hardware recently~\cite{kirby2023exactefficient}.

In this work, we present the \emph{PauliComposer} (PC) algorithm which significantly expedites this calculation. It exploits the fact that any Pauli word only has one element different from zero per row and column, so a number of calculations can be avoided. Additionally, each matrix entry can be computed without performing any multiplications. Even though the exponential scaling of the Hilbert space cannot be avoided, PC can boost inner calculations where several tensor products involving Pauli matrices appear. In particular, those that appear while building Hamiltonians as weighted sums of Pauli strings or decomposing an operator in the Pauli basis.

The PC algorithm could be implemented in computational frameworks in which this sort of operations are crucial, such as the Python modules Qiskit~\cite{qiskit}, PennyLane~\cite{pennylane}, OpenFermion~\cite{openfermion} and Cirq~\cite{cirq}. It can also potentially be used in many other applications, such as the Pauli basis decomposition of the Fock space~\cite{entropy} and conventional computation of Ising model Hamiltonians to solve optimization problems~\cite{Lucas_2014, VRP_TSP, BPP, BPP2}, among others.

The rest of the article is organized as follows: in \seclabel{sec:desc} we describe the algorithm formulation in depth, showing a pseudocode-written routine for its computation. In \seclabel{sec:bench}, a set of tests is performed to show that a remarkable speed-up can be achieved when compared to state-of-the-art techniques. In \seclabel{sec:uses}, we show how this PC algorithm can be used to solve relevant problems. Finally, the conclusions drawn from the presented results are given in \seclabel{sec:conclusions}. We provide proofs for several statements and details of the algorithm in the appendices.

\section{Algorithm formulation}\label{sec:desc}

In this section we discuss the PC algorithm formulation in detail. Pauli matrices are hermitian, involutory and unitary matrices that together with the identity form the set $\sigma_{\{0, 1, 2, 3\}} = {\{I,X, Y, Z\}}$. Given an input string $x = x_{n-1}\dots x_0\in\{0,1,2,3\}^n$, the PC algorithm constructs
\begin{equation}\label{eq:composition}
P(x) \coloneqq \sigma_{x_{n-1}}\otimes\sigma_{x_{n-2}}\otimes\dots\otimes\sigma_{x_0}.
\end{equation}

Let us denote its matrix elements as $P_{j,k}(x)$ with $j, k=0,\dots,2^n-1$. It is important to remark that for each row $j$, there will be a single column $k(j)$ such that $P_{j, k(j)} \neq 0$ (see~\applabel{sec:zeros}).
The solution amounts to a map from the initial Pauli string to the positions and values of the $2^n$ nonzero elements. This calculation will be done sequentially, hence the complexity of the algorithm will be bounded from below by this number.

As a first step, it is worth noting that Pauli string matrices are either real (all elements are $\pm1$) or purely imaginary (all are $\pm i$). This depends on $n_Y$, the number of $Y$ operators in $P(x)$. We can redefine $\tilde{Y}\coloneqq iY$, so that $\smash[t]{\tilde{\sigma}_{\{0,1,2,3\}}=\{I, X, \tilde{Y}, Z\}}$ and $\smash[t]{\tilde{P}(x)\coloneqq \tilde{\sigma}_{x_{n-1}}\otimes\dots\otimes\tilde{\sigma}_{x_0}}$. As a result, every entry in $\smash[t]{\tilde{P}(x)}$ will be $\pm 1$. This implies that there is no need to compute any multiplication: the problem reduces to locating the nonzero entries in $\smash[t]{\tilde{P}(x)}$ and tracking sign changes. The original $P(x)$ can be recovered as $\smash[t]{P(x)=(-i)^{n_Y\text{ mod }4}\tilde{P}(x)}$.

We will now present an iterative procedure to compute $\tilde{P}$ by finding for each row $j$ the nonzero column number $k(j)$ and its corresponding value $\smash[t]{\tilde{P}_{j,k(j)}}$. For the first row, $j=0$, the nonzero element $\smash[t]{\tilde{P}_{0,k(0)}}$, can be found at
\begin{equation}\label{eq:first_column}
    k(0)= [y(x_{n-1})\dots y(x_{0})]_{10},
\end{equation}
where $[\,\cdot\,]_{10}$ is the decimal representation of a bit string and $y(x_i)$ tracks the diagonality of $\sigma_{x_i}$, where $y(x_i)$ is equal to $0$ if $x_i=\{0,3\}$ (thus $\sigma_{x_i}\in\{I,Z\}$) and $1$ otherwise (thus $\sigma_{x_i}\in\{X,Y\}$). The value of this entry is
\begin{equation}\label{eq:first_entry}
    \tilde{P}_{0, k(0)}=+1 \implies P_{0, k(0)}= (-i)^{n_Y\text{ mod }4}.
\end{equation}

The following entries can be computed iteratively. At the end of stage $l$, with $l=0, \cdots, n-1$, all nonzero elements in the first $2^{l+1}$ rows of $\smash[t]{P_{j,k(j)}}$ will have been computed using the information given by the substring $x_l \dots x_0$. At the next step, $l+1$, the following $2^l$ rows are filled using the ones that had already been computed, where the row-column relation $k(j)$ is given by
\begin{equation}\label{eq:columns}
    k(j+2^l)=k(j) + (-1)^{y(x_l)} 2^l, \quad j=0, \dots, 2^l-1.
\end{equation}
The second term of the RHS of this relation takes into account the way that the blocks of zeros returned at stage $l$ affect the new relative location of the nonzero blocks within the new $2^{l+1}\times2^{l+1}$ subcomposition. Its corresponding values are obtained from the previous ones, up to a possible change of sign given by
\begin{equation}\label{eq:entries}
    P_{j+2^l, k(j+2^l)} = \epsilon_l P_{j, k(j)},
\end{equation}
with $\epsilon_l$ equal to $1$ if $x_l\in\{0,1\}$ and $-1$ otherwise. This $\epsilon_l$ is nothing but a parameter that takes into account if $\sigma_{x_l}$ introduces a sign flip. In~\alglabel{alg:code} a pseudocode that summarizes the presented algorithm using~\eqref{eq:first_column}-\eqref{eq:entries}, is shown.

For the particular case of diagonal Pauli strings (only $I$ and $Z$ matrices), there is no need to compute the row-column relation $k(j)$, just the sign assignment is enough. Even if this is also the case for anti-diagonal matrices, we focus on the diagonal case due to its relevance in combinatorial problems~\cite{Lucas_2014, VRP_TSP, BPP, BPP2}. See~\alglabel{alg:code_diag} for the pseudocode of this case (\texttt{PDC} stands for \emph{PauliDiagonalComposer}).

\begin{figure}[!t]
 \begin{minipage}[t]{\columnwidth}
  \begin{algorithm}[H]
\DontPrintSemicolon
\SetAlgoNoLine
\SetKwInOut{Input}{input}
\SetKwInOut{Output}{output}
\SetKwFunction{Len}{len}
\SetKwFunction{Range}{range}
\Input{$x_{n-1}x_{n-2}\dots x_0\leftarrow$ string with $x_i\in\{0,1,2,3\}$}
$n\leftarrow\Len{x}$\;
$n_Y\leftarrow\,$number of $Y$ matrices in $x$\;
$j\leftarrow\Range{$0, 2^n-1$}$\tcp*[r]{rows}
$k, m\leftarrow\,$empty $2^n$-array\tcp*[r]{columns/entries}
$k(0)\leftarrow y(x_{n-1})\dots y(x_{0})$ in base 10\;
$m(0)\leftarrow(-i)^{n_Y\text{ mod }4}$\;\label{line:const}
\For{$l\in$ \Range{$0, n-1$}}{
    $k(2^l:2^{l+1}-1)\leftarrow k(0:2^l-1)+(-1)^{y(x_l)}2^l$\;
    \uIf(\hfill\tcp*[f]{$\epsilon_l=1$}){$x_l\in\{0, 1\}$}{
        $m(2^l:2^{l+1}-1)\leftarrow m(0:2^l-1)$
    }
    \uElse(\hfill\tcp*[f]{$\epsilon_l=-1$}){
        $m(2^l:2^{l+1}-1)\leftarrow -m(0:2^l-1)$\;
    }
}
\Output{$P(x)$ as a sparse matrix stacking $(j, k, m)$}
\caption{\small\texttt{PC}: compose $n$ Pauli matrices}\label{alg:code}
\end{algorithm}
\end{minipage}\vspace{-1.8mm}
\begin{minipage}[t]{\columnwidth}
\begin{algorithm}[H]
\DontPrintSemicolon
\SetAlgoNoLine
\SetKwInOut{Input}{input}
\SetKwInOut{Output}{output}
\SetKw{Continue}{continue}
\SetKwFunction{Len}{len}
\SetKwFunction{Range}{range}
\Input{$x_{n-1}x_{n-2}\dots x_0\leftarrow$ string with $x_i\in\{0,3\}$}
$n\leftarrow\Len{x}$\;
$j, k\leftarrow\Range{$0, 2^n-1$}$\tcp*[r]{rows/columns}
$m\leftarrow\,$empty $2^n$-array \hfill\tcp*[r]{entries}
$m(0)\leftarrow 1$\;\label{line:const_diag}
\For{$l\in$ \Range{$0, n-1$}}{
    \uIf(\hfill\tcp*[f]{$\epsilon_l=1$}){$x_l=0$}{
        $m(2^l:2^{l+1}-1)\leftarrow m(0:2^l-1)$
    }
    \uElse(\hfill\tcp*[f]{$\epsilon_l=-1$}){
        $m(2^l:2^{l+1}-1)\leftarrow -m(0:2^l-1)$\;
    }
}
\Output{$P(x)$ as a sparse matrix stacking $(j, k, m)$}
\caption{\small\texttt{PDC}: compose $n$ diagonal Pauli matrices}\label{alg:code_diag}
\end{algorithm}%
 \end{minipage}
\end{figure}%

The PC algorithm is able to circumvent the calculation of a significant amount of operations. When generic Kronecker product routines (see~\applabel{sec:methods}) are used for the same task, the amount of multiplications needed for computing a Pauli string is $\mathcal{O}[n2^{2n}]$ and $\mathcal{O}[n2^n]$ for dense and sparse matrices, respectively. In contrast, the PC algorithm, considering the worst-case scenarios, needs
\begin{itemize}
    \item $\{I, Z\}^{\otimes n}$: $\mathcal{O}[2^n]$ changes of sign.
    \item Otherwise: $\mathcal{O}[2^n]$ sums and $\mathcal{O}[2^n]$ changes of sign.
\end{itemize}
In all cases our algorithm can significantly outperform those that are not specifically designed for Pauli matrices.

On top of that, this method is also advantageous for computing weighted Pauli strings. Following~\eqref{eq:first_entry}, $W\coloneqq\omega P$, with arbitrary $\omega$, can be computed by defining $W_{0, k(0)}= \omega(-i)^{n_Y\text{ mod }4}$ which avoids having to do any extra multiplication. This change is reflected in~\alglabel{alg:code} by changing line~\ref*{line:const} to $m(0)\gets \omega(-i)^{n_Y\text{ mod }4}$ and line~\ref*{line:const_diag} to $m(0)\gets\omega$ in~\alglabel{alg:code_diag}. This is specially important as it can be used to compute Hamiltonians written as a weighted sum of Pauli strings, where $H=\sum_x \omega_xP(x)$.

\section{Benchmarking}\label{sec:bench}

In this section we analyse the improvement that the PC strategy introduces against other known algorithms labelled as \texttt{Naive} (regular Kronecker product), Algorithm 993 (\texttt{Alg993})~\cite{Fackler_2019}, \texttt{Mixed} and \texttt{Tree}~\cite{topics_in_matrix_analysis, kron_efficiently}. Further details can be found in~\applabel{sec:methods}.
We benchmark these algorithms using MATLAB~\cite{MATLAB:R2022a} as it is proficient at operating with matrices (it incorporates optimized routines of the well-known BLAS and LAPACK libraries~\cite{Lawson_79, lapack}). The PC avoids matrix operations and thus it would not be ideal to implement it using MATLAB. Instead, we use Python~\cite{python39} since many quantum computing libraries are written in this language~\cite{qiskit,cirq,openfermion,pennylane}. See~\tablabel{tab:specs} for a full description of the computational resources used.%
\begin{table}[!tb]
\caption{Computer specifications.}\label{tab:specs}%
\begin{ruledtabular}
\begin{tabular}{rlrl}
        Processor & \multicolumn{3}{l}{Intel Core i7-11850H ($16\times\SI{2.50}{GHz}$)} \\
        RAM & \multicolumn{3}{l}{$\SI{32.0}{GB}$ (DDR4)} \\
        OS & \multicolumn{3}{l}{Ubuntu 22.04.1 LTS ($\times$64)} \\
        MATLAB~\cite{MATLAB:R2022a} & \multicolumn{3}{l}{9.12.0.1884302 (R2022a)} \\
        Python~\cite{python39} & \multicolumn{3}{l}{3.9.12} \\
        NumPy~\cite{numpy} & 1.23.2 & SciPy~\cite{2020SciPy-NMeth} & 1.9.0 \\
        Qiskit~\cite{qiskit} & 0.38.0 & PennyLane~\cite{pennylane} & 0.23.1 \\
\end{tabular}
\end{ruledtabular}
\end{table}%
\begin{figure}[!t]
 \centering
 \includegraphics{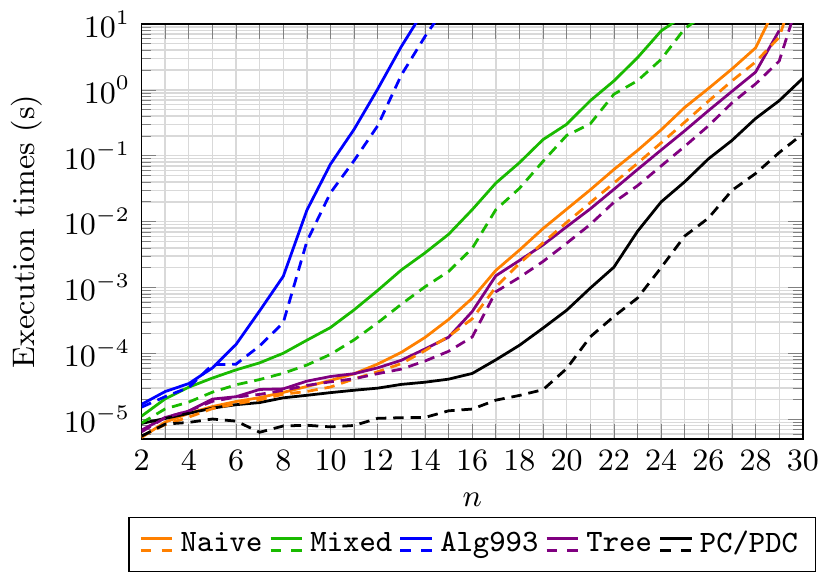}%
\caption{(Color online) Execution times for computing general (solid line) and diagonal (dashed) $n$-Pauli strings using different methods.}\label{fig:times}
\end{figure}%

Concerning memory needs, with this algorithm only $2^n$ nonzero elements out of $2^{2n}$ are stored. This is exactly the same as using sparse matrices, thus, no major improvement is to be expected. As for the computational time, we compare how different algorithms behave as the length $n$ of the Pauli string increases. In~\figlabel{fig:times} execution times for general and diagonal Pauli strings are shown. For the PC methods, we use the \texttt{PC} routine (\alglabel{alg:code}) for the general case and the \texttt{PDC} routine (\alglabel{alg:code_diag}) for the diagonal one. In accordance to our theoretical analysis, the PC algorithm proves to be the best performing routine.

On a more technical note, when using the PC routine, matrices with complex values ($n_Y$ odd) take twice as much time as real valued ones ($n_Y$ even). Consequently, we compute their execution times separately and then average them. Moreover, it is convenient to choose when to use \texttt{PC} or \texttt{PDC} as the latter can be up to 10 times faster.

\section{Real use cases of the PauliComposer algorithm}\label{sec:uses}
The PC algorithm can be used to perform useful calculations in physics. In this section, the Pauli basis decomposition of a Hamiltonian and the construction of a Hamiltonian as a sum of weighted Pauli strings are discussed in detail. Another worth mentioning scenario is the digital implementation of the complex exponential of a Pauli string, i.e. $e^{-i \theta P(x)}=\cos(\theta)I -i \sin(\theta)P(x)$.

\emph{Pauli basis decomposition of a Hamiltonian}.---The decomposition of a Hamiltonian written as a $2^n\times2^n$ matrix into the Pauli basis is a common problem in quantum computing. Given a general Hamiltonian $H$, this decomposition can be written as $H=\sum_x \omega_x P(x)$ with $x = x_{n-1}\dots x_0$ and $P(x)$ as in \eqref{eq:composition}. The coefficients $\omega_x$ are obtained from the orthogonal projection as
\begin{equation}\label{eq:coeffs}
    \omega_x = \frac{1}{2^n} \trace{P(x) H} =  \frac{1}{2^n}\sum_{j=0}^{2^n-1}P_{j,k(j)}(x) H_{k(j),j}.
 \end{equation}
 Following the discussion in \seclabel{sec:desc}, the double sum collapses to a single one in~\eqref{eq:coeffs} since there is only one nonzero element per row and column. Each of these weights can be computed independently, which allows for a parallel implementation. Additionally, in some special cases, it can be known in advance if some $\omega_x$ will vanish:
\begin{itemize}
\item If $H$ is symmetric, strings with an odd number of $Y$ matrices can be avoided ($2^{n-1}(2^n+1)$ terms).
\item If $H$ is diagonal, only strings composed by $I$ and $Z$ will contribute ($2^n$ terms).
\end{itemize}

\begin{table*}[!t]
\centering
\caption{Execution times (in seconds) for decomposing an arbitrary $2^n\times2^n$ matrix. Here, \texttt{PC} and \texttt{PDC} calculations were made computing weights sequentially and in parallel.}\label{tab:decomposer_times}
\begin{ruledtabular}
\begin{tabular}{@{}lddddddddd@{}}
\multicolumn{1}{c}{$n$} & \multicolumn{1}{c}{2} & \multicolumn{1}{c}{3} & \multicolumn{1}{c}{4} & \multicolumn{1}{c}{5} & \multicolumn{1}{c}{6} & \multicolumn{1}{c}{7} & \multicolumn{1}{c}{8} & \multicolumn{1}{c}{9} & \multicolumn{1}{c}{10} \\ \midrule
\multicolumn{10}{c}{Non-hermitian matrix $H_\text{NH}$} \\ \midrule
\texttt{PC} (sequential) & 0.0005 & 0.0021 & 0.012 & 0.078 & 0.55 & 4.06 & 31.2 & 254 & 2008 \\
\texttt{PC} (parallel) & 0.094 & 0.093 & 0.11 & 0.15 & 0.38 & 2.10 & 13.5 & 94.3 & 719 \\
Qiskit & 0.0015 & 0.0050 & 0.020 & 0.14 & 1.16 & 8.78 & 92.38 & 1398 & 26938 \\ \midrule
\multicolumn{10}{c}{Hermitian matrix $H_\text{H}$} \\ \midrule
\texttt{PC} (sequential) & 0.0004 & 0.0021 & 0.012 & 0.078 & 0.56 & 4.24 & 32.86 & 261 & 2007 \\
\texttt{PC} (parallel) & 0.068 & 0.070 & 0.079 & 0.12 & 0.33 & 1.99 & 13.02 & 96.5 & 647 \\
Qiskit & 0.0010 & 0.0035 & 0.018 & 0.10 & 1.47 & 12.02 & 108 & 1295 & 26848 \\
PennyLane & 0.0013 & 0.0060 & 0.030 & 0.15 & 2.23 & 10.66 & 97.6 & 2019 & 35014 \\ \midrule
\multicolumn{10}{c}{Symmetric matrix $H_\text{S}$} \\
\midrule
\texttt{PC} (sequential) & 0.0003 & 0.0010 & 0.0058 & 0.036 & 0.24 & 1.78 & 14.05 & 108 & 794 \\
\texttt{PC} (parallel) & 0.059 & 0.059 & 0.061 & 0.078 & 0.13 & 0.48 & 2.75 & 20.1 & 140 \\
Qiskit & 0.0010 & 0.0036 & 0.018 & 0.10 & 1.45 & 11.07 & 105 & 1320 & 26399 \\
PennyLane & 0.0011 & 0.0054 & 0.027 & 0.13 & 1.36 & 9.22 & 91.52 & 1477 & 31583 \\ \midrule
\multicolumn{10}{c}{Diagonal matrix $H_\text{D}$} \\
\midrule
\texttt{PDC} (sequential) & 0.0001 & 0.0002 & 0.0006 & 0.0018 & 0.0068 & 0.025 & 0.094 & 0.37 & 1.49 \\
\texttt{PDC} (parallel) & 0.055 & 0.057 & 0.059 & 0.060 & 0.060 & 0.064 & 0.078 & 0.12 & 0.35 \\
Qiskit & 0.0010 & 0.0035 & 0.018 & 0.10 & 1.46 & 11.0 & 103 & 1270 & 25977 \\
PennyLane & 0.0010 & 0.0047 & 0.023 & 0.11 & 1.20 & 8.29 & 86.2 & 1370 & 30941 \\
\end{tabular}
\end{ruledtabular}
\end{table*}%
\begin{figure*}[!tb]
 \centering
 \includegraphics{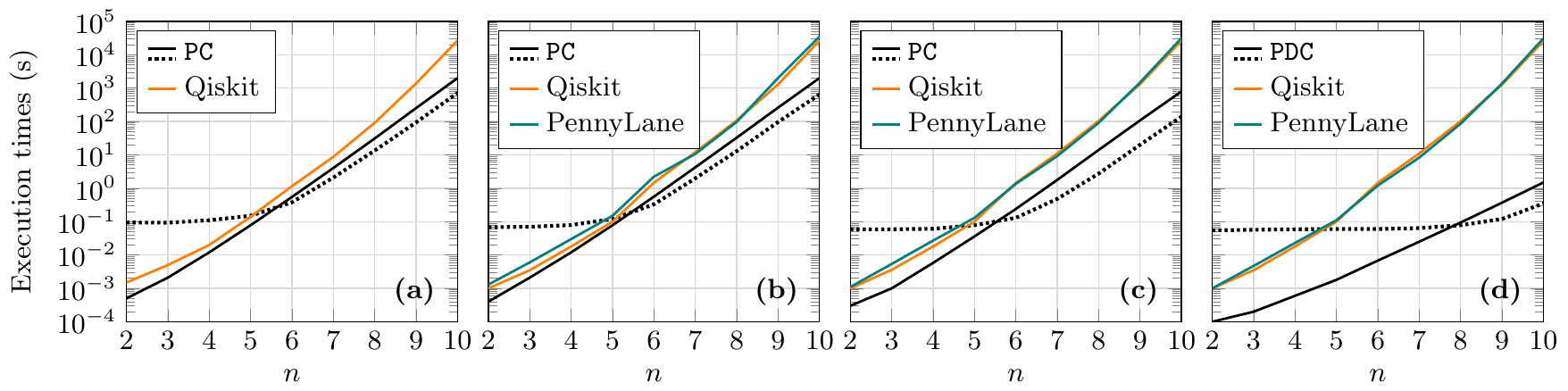}%
\caption{(Color online) Execution times for decomposing $2^n\times 2^n$ \textbf{(a)} non-hermitian $H_\text{NH}$, \textbf{(b)} hermitian $H_\text{H}$, \textbf{(c)} symmetric $H_\text{S}$ and \textbf{(d)} diagonal $H_\text{D}$ matrices with different methods. For \texttt{PC} and \texttt{PDC}, solid (dotted) line depicts sequential (parallelized) decomposition. See~\tablabel{tab:decomposer_times}. As expected, notice that the larger $n$, the higher impact of parallelization.}\label{fig:times_decomposer}
\end{figure*}%

The operations made by \emph{PauliDecomposer} (PD) are
\begin{itemize}
    \item If $H$ is diagonal ($\mathcal{O}[2^n]$ strings): $\mathcal{O}[2^{2n}]$ operations.
    \item Otherwise ($\mathcal{O}[2^{2n}]$ strings): $\mathcal{O}[2^{3n}]$ operations.
\end{itemize}%
This PD algorithm checks if the input matrix satisfies one of the aforementioned cases and computes the coefficients using the PC routine and~\eqref{eq:coeffs}, discarding all vanishing Pauli strings. This workflow considerably enhances our results, especially for diagonal matrices.

In~\tablabel{tab:decomposer_times} and~\figlabel{fig:times_decomposer}, we tested the most extended methods for decomposing matrices into weighted sums of Pauli strings against PD, using Python~\cite{python39} to compare their performance. In particular, we used the \texttt{SparsePauliOp} class from Qiskit~\cite{qiskit} and the \texttt{decompose\_hamiltonian} function from PennyLane~\cite{pennylane} (only works with hermitian Hamiltonians). To the best of authors' knowledge, both routines are based on \texttt{Naive} approach without inspecting the input matrix nature before proceeding.

Four types of random $2^n\times2^n$ matrices were generated, namely non-hermitian $H_\text{NH}$, hermitian $H_\text{H}$, symmetric $H_\text{S}$ and diagonal $H_\text{D}$ matrices. The PD vastly outperforms Qiskit and PennyLane routines, specially for the symmetric and diagonal cases.

\emph{Building of a Hamiltonian as a sum of weighted Pauli strings}.---Many Hamiltonians are written in terms of weighted Pauli strings. Our method can compute weighted Pauli strings directly without extra computations. In~\figlabel{fig:string_constant} we show a performance comparison of the presented methods for computing Hamiltonians written as sums of weighted Pauli strings. The Hamiltonian used is similar to the one proposed in~\cite{BPP2},
\begin{equation}\label{eq:ham_bpp}
    H = \sum_{i=0}^{n-1}\alpha_i\sigma^i_3 + \sum_{i<j}^{n-1}\beta_{ij}\sigma^i_3\sigma^j_3,
\end{equation}
being the corresponding weights $\vec{\alpha}$ and $\vec{\beta}$ arbitrary and $\sigma^i_3$ as defined in~\eqref{eq:tensor_product_decompose_backward}. This Hamiltonian is computed using~\alglabel{alg:ising_ham}, which uses the \texttt{PDC} routine (see \alglabel{alg:code_diag}) with two inputs: the string $x\in\{0,3\}^n$ to compute and the weights to consider. In the PDC case, we use two strategies: compute each weighted term of~\eqref{eq:ham_bpp} directly and compute each Pauli string and then multiply it by its corresponding weight (solid and dashed lines in~\figlabel{fig:string_constant}, respectively). This is done by changing lines~\ref*{line:str1} to $H\gets H+\alpha_i\texttt{PDC(}str_1\texttt{)}$ and~\ref*{line:str2} to $H\gets H+\beta_{ij}\texttt{PDC(}str_2\texttt{)}$ in~\alglabel{alg:ising_ham} for the second one. There is no significant difference between both methods.

\begin{algorithm}[!t]
\DontPrintSemicolon
\SetAlgoNoLine
\SetKwInOut{Input}{input}
\SetKwInOut{Output}{output}
\SetKwFunction{Len}{len}
\SetKwFunction{Range}{range}
\SetKwFunction{Copy}{copy}
\SetKwFunction{PDC}{PDC}
\SetKwFunction{In}{in}
\Input{$\vec{\alpha},\vec{\beta}\gets\,$lists of weights}
$n\gets\Len{$\vec{\alpha}$}$\;
$H\gets\,2^n\times2^n$ sparse matrix of zeros\;
\For{$i\in$ \Range{$0, n-1$}}{
    $str_1\gets\,$string of $n$ zeros\tcp*[r]{$n$ identities}
    $str_1(i)\gets3$\tcp*[r]{$Z$ in the $i$-th position}
    $H\gets H+\PDC{$str_1, \alpha_i$}$\;\label{line:str1}
    \For{$j\in$ \Range{$i+1, n-1$}}{
        $str_2\gets\Copy{$str_1$}$\;
        $str_2(j)\gets3$\tcp*[r]{$Z$ in the $j$-th position}
        $H\gets H+\PDC{$str_2, \beta_{ij}$}$\;\label{line:str2}
    }
}
\Output{Hamiltonian $H$ as a sparse matrix}
\caption{\small Ising model Hamiltonian computation}\label{alg:ising_ham}
\end{algorithm}
\begin{figure}[!tb]
 \centering
 \includegraphics{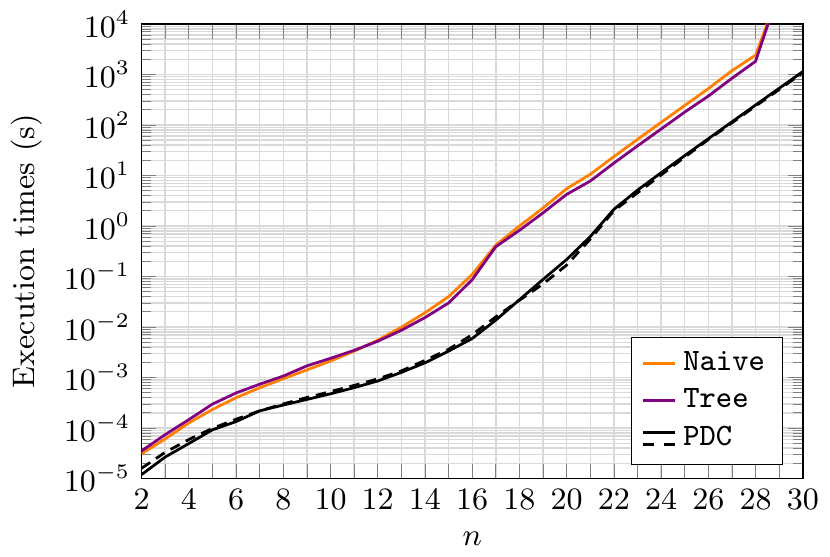}%
\caption{(Color online) Execution times for computing~\eqref{eq:ham_bpp} using~\alglabel{alg:ising_ham} (solid line) and computing previously the Pauli string and multiply it by its corresponding weight (dashed).}\label{fig:string_constant}
\end{figure}%

\section{Conclusions}\label{sec:conclusions}

The fast and reliable computation of tensor products of Pauli matrices is crucial in the field of quantum mechanics and, in particular, of quantum computing. In this article we propose a novel algorithm with proven theoretical and experimental enhancements over similar methods of this key yet computationally tedious task. This is achieved by taking advantage of the properties of Pauli matrices and the tensor product definition, which implies that one can avoid trivial operations such as multiplying constants by one and waste time computing elements with value zero that could be known in advance.

Concerning memory resources, it is convenient to store the obtained results as sparse matrices since only $2^n$ out of $2^{2n}$ entries will not be zero for a Pauli string of length $n$, i.e. the density of the resultant matrix will be $2^{-n}$ (see~\applabel{sec:zeros}).

Our benchmark tests suggest that the PauliComposer algorithm and its variants can achieve a remarkable acceleration when compared to the most well-known methods for the same purpose both for single Pauli strings and real use cases. In particular, the most considerable outperformance can be seen in~\tablabel{tab:decomposer_times} for the symmetric and diagonal matrix decomposition over the Pauli basis.

Finally, its simple implementation (\alglabel{alg:code}-\ref{alg:code_diag}) can potentially allow to integrate the PC routines into quantum simulation packages to enhance inner calculations.

\begin{acknowledgments}
We thank Javier Mas Solé, Yue Ban and Mikel García de Andoin for the helpful discussions. This research is funded by the project “BRTA QUANTUM: Hacia una especialización armonizada en tecnologías cuánticas en BRTA” (expedient no. KK-2022/00041). The work of JSS has received support from Xunta de Galicia (Centro singular de investigación de Galicia accreditation 2019-2022) by European Union ERDF, from the Spanish Research State Agency (grant PID2020-114157GB-100) and from MICIN with funding from the European Union NextGenerationEU (PRTR-C17.I1) and the Galician Regional Government with own funding through the “Planes Complementarios de I+D+I con las Comunidades Autónomas” in Quantum Communication.
\end{acknowledgments}
\noindent\textbf{Data and code availability statement.} The data used in the current study is available upon reasonable request from the corresponding authors. The code used can be found at \url{https://github.com/sebastianvromero/PauliComposer}.

\appendix
\section{Some proofs regarding Pauli strings}\label{sec:zeros}
In this section we prove two key properties of Pauli strings on which our algorithm is based.
\begin{theorem}\label{th:zeros}
 A Pauli string $P(x)$ of length $n$ given by~\eqref{eq:composition} has only $2^n$ nonzero entries.
\end{theorem}
\begin{proof}
With the help of~\figlabel{fig:zeros}, we can compute the number of zeros in the resulting matrix as
\begin{equation}\label{eq:zeros_proof_2}
 \begin{aligned}
  n_0(n) &= \textcolor{red}{2\left(2^{n-1}\times2^{n-1}\right)} + \textcolor{green!50!black}{4\left(2^{n-2}\times2^{n-2}\right)} \\
   &+ \textcolor{blue}{8\left(2^{n-3}\times2^{n-3}\right)} + \dots + 2^n(1\times1) \\
   & = \sum_{k=n}^{2n-1}2^k = 2^n\left(2^n-1\right).
 \end{aligned}
\end{equation}
In other words, $P(x)$ will have only $2^n$ nonzero terms. We can prove~\eqref{eq:zeros_proof_2} by induction easily: since $n_0(n=1)$ is true, if we assume that $n_0(n)$ holds we can see that
\begin{equation}
 \begin{aligned}
  n_0(n+1) &= 2\cdot 2^n(2^n-1) + 2\cdot 2^{2n} = 2^{n+1}\left(2^{n+1}-1\right)
 \end{aligned}
\end{equation}
also holds true.
\end{proof}%

\begin{figure}[!bt]
\[
  \bigotimes_{i=0}^{n-1}\sigma_{x_{n-i-1}} =
  \begin{bmatrix}
    \begin{array}{cc}
    \textcolor{green!50!black}{0} & \begin{array}{cc}\cdots&\textcolor{blue}{0}\\\textcolor{blue}{0}&\cdots\end{array} \\
    \begin{array}{cc}\cdots&\textcolor{blue}{0}\\\textcolor{blue}{0}&\cdots\end{array} & \textcolor{green!50!black}{0} \\
    \end{array} & \textcolor{red}{0} \\
    \textcolor{red}{0} & \begin{array}{cc}
    \textcolor{green!50!black}{0} & \begin{array}{cc}\cdots&\textcolor{blue}{0}\\\textcolor{blue}{0}&\cdots\end{array} \\
    \begin{array}{cc}\cdots&\textcolor{blue}{0}\\\textcolor{blue}{0}&\cdots\end{array} & \textcolor{green!50!black}{0} \\
    \end{array} \\
  \end{bmatrix}
\]
\caption{(Color online) Scheme for computing the number of zeros of an arbitrary composition of $n$ Pauli matrices.}\label{fig:zeros}
\end{figure}

\begin{corollary}
 A Pauli string $P(x)$ of length $n$ given by~\eqref{eq:composition} has only one nonzero entry per row and column.
\end{corollary}%
\begin{proof}
 Since the tensor product of unitary matrices is also unitary, then $|\!\det P(x)|=1$. From~\theolabel{th:zeros}, only $2^n$ entries of the resulting $2^n\times2^n$ matrix are nonzero. So the logical conclusion to be drawn is that the unique way to locate them without having a row and a column full of zeros, thus returning a zero determinant, is that each row and column must have only one nonzero entry.
\end{proof}%

\section{Standard methods for computing tensor products}\label{sec:methods}

In this appendix we briefly review the well established algorithms that were used in the benchmark~\cite{Fackler_2019, topics_in_matrix_analysis, kron_efficiently}. First, one can consider what we call the \texttt{Naive} algorithm, which consists on performing the calculations directly. It is clearly highly inefficient as it scales in the number of operations as $\mathcal{O}[n2^n]$ for sparse Pauli matrices.
Second, the \texttt{Mixed} algorithm uses the mixed-product property
\begin{equation}
    \bigotimes_{i=0}^{n-1}\sigma_{x_{n-i-1}}=\prod_{i=0}^{n-1}\sigma^i_{x_{n-i-1}},
\end{equation}
with
\begin{equation}\label{eq:tensor_product_decompose_backward}
    \sigma^i_{x_i}\coloneqq
    \begin{cases}
    I^{\otimes n-1}\otimes\sigma_{x_0} & \text{if }i=0 \\
    I^{\otimes n-i-1}\otimes\sigma_{x_i}\otimes I^{\otimes i} & \text{if }0<i<n-1 \\
    \sigma_{x_{n-1}}\otimes I^{\otimes n-1} & \text{if }i=n-1
    \end{cases},
\end{equation}
to simplify the calculation into a simple product of block diagonal matrices. Based on this procedure, \texttt{Alg993} is presented in~\cite{Fackler_2019}. It can be shown that this method performs over $\mathcal{O}[n2^n]$ operations. Besides that, as~\figlabel{fig:times} suggests, the fact that it requires to transpose and reshape several matrices has a non-negligible effect that fatally increases its computation time. Finally, the \texttt{Tree} routine starts storing pairs of tensor products as
\begin{equation}
    \begin{aligned}
        \displaystyle\left\{\sigma_{x_{n-2i-1}}\otimes\sigma_{x_{n-2i-2}}\right\}^{n/2-1}_{i=0}& &&\text{if }n\text{ even} \\
        \displaystyle\left\{\sigma_{x_{n-1}}\right\} \cup \left\{\sigma_{x_{n-2i-1}}\otimes\sigma_{x_{n-2i-2}}\right\}^{\lfloor n/2\rfloor}_{i=0}& &&\text{if }n\text{ odd}
    \end{aligned},
\end{equation}
and proceeds with the resultant matrices following the same logic, which allows to compute~\eqref{eq:composition} by iteratively grouping its terms by pairs. For better results, this method can be parallelized.

\bibliography{bibfile}

\end{document}